\documentclass[11pt]{article}
\usepackage{fullpage}
\usepackage{rpmacros}
\RequirePackage[colorlinks=true]{hyperref}
\hypersetup{
  linkcolor=[rgb]{0.3,0.3,0.6},
  citecolor=[rgb]{0.2, 0.6, 0.2},
  urlcolor=[rgb]{0.6, 0.2, 0.2}
}
\usepackage{mathpazo}
\usepackage{bbm}
\usepackage{todonotes}
\usepackage{lipsum}
\usepackage{setspace}

\usepackage{amsthm}
\usepackage{thmtools,thm-restate}

\numberwithin{equation}{section}
\declaretheoremstyle[bodyfont=\it,qed=\qedsymbol]{noproofstyle}

\declaretheorem[numberlike=equation]{observation}

\declaretheorem[name=Observation,numbered=no]{observation*}

\declaretheorem[numberlike=equation]{theorem}

\declaretheorem[name=Theorem,numbered=no]{theorem*}

\declaretheorem[numberlike=equation]{lemma}
\declaretheorem[name=Lemma,numbered=no]{lemma*}

\declaretheorem[numberlike=equation]{corollary}
\declaretheorem[name=Corollary,numbered=no]{corollary*}

\declaretheorem[name=Proposition,numbered=no]{proposition*}

\declaretheorem[numberlike=equation]{claim}
\declaretheorem[name=Claim,numbered=no]{claim*}

\declaretheorem[name=Conjecture,numbered=no]{conjecture*}

\declaretheorem[name=Question,numbered=no]{question*}

\declaretheoremstyle[bodyfont=\it,qed=$\lozenge$]{defstyle} 

\declaretheorem[numberlike=equation,style=defstyle]{definition}
\declaretheorem[unnumbered,name=Definition,style=defstyle]{definition*}

\declaretheorem[unnumbered,name=Example,style=defstyle]{example*}

\declaretheorem[unnumbered,name=Notation=defstyle]{notation*}

\declaretheorem[unnumbered,name=Construction,style=defstyle]{construction*}

\declaretheorem[unnumbered,name=Remark,style=defstyle]{remark*}

\def\dim{\mathsf{Dim}}
\def\eval{\mathsf{Eval}}
\def\span{\mathsf{Span}}

\def\deg{\mathsf{Deg}}

\def\NW{\mathsf{NW}}
\def\NWL{\mathsf{NW \circ Lin}}

\def\Perm{\mathsf{Perm}}
\def\Det{\mathsf{Det}}
\def\acc{\mathsf{ACC^0}}
\def\sp{\mathsf{\#P}}
\def\nexp{\mathsf{NEXP}}
\def\acp{\mathsf{AC^0[p]}}
\usepackage{nth}
\usepackage{intcalc}
\usepackage{etoolbox}
\usepackage{xstring}
\hypersetup{
breaklinks=true 
}

\newcommand{\ECCCurl}[2]{http://eccc.hpi-web.de/report/\ifnumcomp{#1}{>}{93}{19}{20}#1/#2/}

\newcommand{\shortECCC}[2]{\texttt{\href{http://eccc.hpi-web.de/report/\ifnumcomp{#1}{>}{93}{19}{20}#1/#2/}{eccc:TR#1-#2}}}

\newcommand{\parseECCC}[1]{
\StrSubstitute{#1}{TR}{}[\tmpstring]%
\IfSubStr{\tmpstring}{/}{ 
\StrBefore{\tmpstring}{/}[\ecccyear]%
\StrBehind{\tmpstring}{/}[\ecccreport]%
}{
\StrBefore{\tmpstring}{-}[\ecccyear]%
\StrBehind{\tmpstring}{-}[\ecccreport]%
}%
\shortECCC{\ecccyear}{\ecccreport}}

\title{Functional lower bounds for arithmetic circuits and connections to boolean circuit complexity}
\author{
Michael A. Forbes\thanks{Department of Computer Science, Princeton University. Research supported by the Princeton Center for Theoretical Computer Science. E-mail :  \texttt{miforbes@csail.mit.edu}}
\and
Mrinal Kumar\thanks{Rutgers University, Research supported in part by the  Simons Graduate Fellowship. Part of this work was done while on  an internship at MSR, New England. E-mail : \texttt{mrinal.kumar@rutgers.edu}
}
\and
Ramprasad Saptharishi\thanks{Tel Aviv University, The research leading to these results has received funding from the European Community's Seventh Framework Programme (FP7/2007-2013) under grant agreement number 257575. E-mail : \texttt{ramprasad@cmi.ac.in}
}
}
\begin{document}
\maketitle

\begin{abstract}
We say that a circuit $C$ over a field $\F$ \emph{functionally} computes an $n$-variate polynomial $P \in \F[x_1, x_2, \ldots, x_n]$ if for every $\vecx \in \{0,1\}^n$ we have that $C(x) = P(x)$. This is in contrast to \emph{syntactically} computing $P$, when $C \equiv P$ as formal polynomials. In this paper, we study the question of proving lower bounds for homogeneous depth-$3$ and depth-$4$ arithmetic circuits for functional computation. We prove the following results : 
\begin{itemize}
\itemsep 0pt
\item Exponential lower bounds homogeneous depth-$3$ arithmetic circuits for a polynomial in $\VNP$. 
\item Exponential lower bounds for homogeneous depth-$4$ arithmetic circuits with bounded individual degree for a polynomial in $\VNP$. 
\end{itemize}
Our main motivation for this line of research comes from our observation that strong enough functional lower bounds for even very special depth-$4$ arithmetic circuits for the Permanent imply a separation between $\sp$ and $\acc$. Thus, improving the second result to get rid of the \emph{bounded individual degree} condition could lead to substantial  progress in boolean circuit complexity. Besides, it is known from a recent result of  Kumar and Saptharishi~\cite{KumarSaptharishi15} that over constant sized finite fields, strong enough \emph{average case} functional lower bounds for homogeneous depth-$4$ circuits imply superpolynomial lower bounds for homogeneous depth-$5$ circuits. 

Our proofs are based on a family of new complexity measures called \emph{shifted evaluation dimension}, and might be of independent interest. \end{abstract}

\thispagestyle{empty}
\newpage
\pagenumbering{arabic}

\section{Introduction}
Arithmetic circuits are one of the most natural models of computation for studying computation with multivariate polynomials. One of the most fundamental questions in this area of research is to show that there are low degree polynomials which cannot be efficiently computed by \emph{small sized} arithmetic circuits. However, in spite of the significance of this question, progress on it has been sparse and our current state of understanding of lower bounds for arithmetic circuits continues to remain extremely modest. 

Most of the research in algebraic complexity theory so far considers arithmetic circuits and multivariate polynomials as \emph{formal} objects and studies the complexity of \emph{syntactic} representation of polynomials  over the underlying field. However, in this work,  we aim to study the \emph{semantic} or \emph{functional} analogue of the complexity of computing multivariate polynomials. We formally define this notion below and then try to motivate the definition based on our potential applications. 
\begin{definition}[Functional equivalence]\label{def: functional equivalence}
 Let $\F$ be any field and let $D$ be a subset of $\F$. We say that two $n$-variate polynomials $P_1$ and $P_2$ in $\F[x_1, x_2, \ldots, x_n]$ are \emph{functionally} equivalent over the domain $D^n$ if 
\[
 \forall \vecx \in D^n\spaced{,} P_1(\vecx) = P_2(\vecx)\qedhere
\]
\end{definition}
This definition of functional equivalence naturally extends to the case of arithmetic circuits functionally computing a family of polynomials, as defined below. 
\begin{definition}[Functional computation]\label{def: functional computation}
 Let $\F$ be any field and let $D$ be a subset of $\F$. A circuit family $\{C_n\}$ is said to functionally compute a family of polynomials $\{P_n\}$ over the domain $D^n$ if 
\[
 \forall n\in \N, \vecx \in D^n\spaced{,} C_n(\vecx) = P_n(\vecx)\qedhere
\]
\end{definition}
Having defined functional computation, we will now try to motivate the problem of proving functional lower bounds for arithmetic circuits. 
\subsection{Motivation}

\paragraph{Improved boolean circuit lower bounds: }
In the late 80's there was some spectacular progress on the question of lower bounds for bounded depth boolean circuits. In particular, Razborov and Smolensky~\cite{smolensky87, razborov87} showed exponential lower bounds for constant depth boolean circuits with AND $(\wedge)$, OR $(\vee)$, Negations $(\neg)$ and $\mod p$ gates for a prime $p$ (i.e the class of $\acp$ circuits). However, the question of proving lower bounds for constant depth boolean circuits which also have $\mod q$ gates for a composite $q$ (i.e the class of general $\acc$ circuits) remained wide open. In general, one major obstacle was that the techniques of Razborov and Smolensky failed for composite moduli, and we could not find alternative techniques which were effective for the problem. 
Although it is widely believed that the the majority function should be hard for such circuits, till a few years ago, we did not even know to show that there is such a language in $\nexp$\footnote{The class of problems in nondeterministic exponential time.}.  In a major breakthrough on this question, Williams~\cite{w11} showed that there is a function in $\nexp$ which requires $\acc$ circuits of superpolynomial size. Along with the result itself, the paper introduced a new proof strategy for showing such lower bounds. However, it still remains wide open to show that there is a function in deterministic exponential time, which requires $\acc$ circuits of superpolynomial size. 

One of our main motivations for studying functional lower bounds for arithmetic circuits is the following lemma which shows that such lower bounds in fairly modest set up would imply a separation between $\sp$ and $\acc$. A formal statement and a simple proof  can be found in \autoref{sec: functional lb and acc}. 

\begin{lemma}[Informal]~\label{lem: acc to functional lb-intro}
Let $\F$ be any field of characteristic zero or at least $\exp\left(\omega\left(\poly(\log n)\right)\right)$. Then, a functional lower bound of $\exp\left(\omega\left(\poly(\log n)\right)\right)$  for the permanent of an $n\times n$ matrix over $\{0,1\}^{n^2}$ 
for depth-$4$ arithmetic circuits with bottom fan-in $\poly(\log n)$ imply that $\sp \neq \acc$.  
\end{lemma}

In fact, we show that something slightly stronger is true. It suffices to prove functional lower bounds for the model of sums of powers of low degree polynomials for the conclusion in \autoref{lem: acc to functional lb-intro} to hold. 

At this point, there are two possible interpretations of the statement of \autoref{lem: acc to functional lb-intro}. For an optimist, it provides another approach to proving new lower bounds for $\acc$, while for a pessimist it points to the fact that the functional lower bounds for depth-$4$ arithmetic circuits could be possibly very challenging. What makes us somewhat optimistic about this strategy is the fact that in the last few years, we seem to have made substantial progress on the question of proving lower bounds for homogeneous depth-$4$ circuits in the syntactic setting \cite{gkks13, FLMS13, KLSS, KS14}. In particular, even though the depth-$4$ circuits obtained in the proof of \autoref{lem: acc to functional lb-intro} are not homogeneous, an exponential lower bound for sums of powers of low degree polynomials is known  in the syntactic set up. Therefore, it makes sense to try and understand if these bounds can be extended to the functional set up as well.

\paragraph{Lower bounds for homogeneous depth-$5$ circuits: }
In a recent work by Kumar and Saptharishi~\cite{KumarSaptharishi15},  it was shown that over constant size finite fields, \emph{average case functional } lower bounds for homogeneous depth-$4$ circuits\footnote{in fact, with bounded bottom fan-in} implies lower bounds for homogeneous depth-$5$ circuits. More precisely, the following lemma was shown:

\begin{lemma}[\cite{KumarSaptharishi15}]\label{lem: avg case to depth 5}
Let $\F_q$ be a finite field such that $q = O(1)$. Let $P$ be a homogeneous polynomial of degree $d$ in $n$ variables over $\F_q$, which can be computed by a  homogeneous depth-$5$ circuit of size at most $O\left( \exp{\left( d^{0.499}\right)} \right)$. Then, there exists a homogeneous depth-$4$ circuit $C'$ of bottom fan-in $O(\sqrt{d})$ and top fan-in at most $O\left( \exp{\left( d^{0.499}\right)} \right)$ such that 
\[
\Pr_{x \in \F_q^n}\left[P(x) \neq C'(x)\right] \leq \exp(-\Omega(\sqrt{d}))\qedhere
\]
\end{lemma}
 
Informally, the lemma shows that over small finite fields strong enough \emph{average case} functional lower bounds for homogeneous depth-$4$ arithmetic circuit with bounded bottom fan-in are sufficient to show superpolynomial lower bounds for homogeneous depth-$5$ circuits. Even though in \cite{KumarSaptharishi15}, the authors do not take this route to eventually prove their lower bounds, 
this connection seems like a strong motivation to study the question of proving functional lower bounds for bounded depth arithmetic circuits.  

\paragraph{Functional lower bounds for bounded depth arithmetic circuits: }
It is immediately clear from the definition that \emph{syntactic} computation implies \emph{functional} computation, 
but vice-versa may not be necessarily true. In this sense, proving lower bounds for functional computation could be potentially 
harder than proving lower bounds for syntactic computation. From this point of view, once we have syntactic lower bounds for a certain class of circuits, it seems natural to ask if these bounds can be extended to the functional framework as well. The last few years have witnessed substantial progress on the question of proving lower bounds for variants of depth-$4$ arithmetic circuits, and in 
this work we explore the question of whether these bounds can be extended to the functional setting. 

\paragraph{Applications to proof complexity lower bounds :}
Functional lower bounds have recently found applications for obtaining lower bounds for algebraic proof systems.  In particular, Forbes, Shpilka, Tzameret, and Wigderson~\cite{FSTW15} have given lower bounds in various algebraic circuit measures for any polynomial agreeing with certain functions of the form $\vecx\mapsto\frac{1}{p(\vecx)}$, where $p$ is a constant-degree polynomial (which is non-zero on the boolean cube). In particular, they used such lower bounds to obtain lower bounds for the various subclasses of the Ideal Proof System (IPS) of Grochow and Pitassi~\cite{GrochowPitassi14}.

In the next section, we explore the connections between syntactic and functional computation in a bit more detail, and discuss why the 
techniques used in proving syntactic lower bounds do not seem conducive to prove lower bounds in the functional setting. Hence, the problem of proving functional lower bounds might lead us to more techniques for arithmetic circuit lower bounds. 
 
 \subsection{Functional vs syntactic computation}
We now discuss the differences and similarities between functional and syntactic computation in a bit more detail. The following observation is easy to see.
\begin{observation}\label{obs: properties of semantic equivalence}
The following properties follow from \autoref{def: functional computation}:
 \begin{itemize}
\item Any two polynomials $P_1$ and $P_2$ which are syntactically equivalent are also functionally equivalent for every choice of $D$.
\item If two polynomials of individual degrees bounded by $d$ are functionally equivalent over any domain of size at least $d+1$, then they are also syntactically equivalent. 
\item In particular, any two multilinear polynomials which are functionally equivalent over the hypercube $\{0,1\}^n$ are also syntactically equivalent. 
\end{itemize}
\end{observation}
For the rest of the paper, our domain of interest will be $D = \{0,1\}$ and we will be interested in polynomials which are functionally the same over the hypercube $\{0,1\}^n$. For brevity, for the rest of the paper, when we say that two polynomials are functionally equivalent, we mean that the domain is  the hypercube. As an additional abuse of notation, when we say that a circuit $C$ is functionally equivalent to a polynomial $P$, we mean that for every $\vecx \in \{0,1\}^n$, $C(\vecx) = P(\vecx)$. Observe that functional equivalence over the hypercube is precisely the same as syntactic equivalence when we work modulo the ideal generated by the polynomials $\{x_i^2-x_i : i \in [n]\}$. However, we find the functional view easier and more convenient to work with.   

At this point, one might ask why is the choice of $D$ as $\{0,1\}$ a natural one? The motivation for studying a domain of size $2$ stems from the fact that most of the polynomials for which we have syntactic arithmetic circuit lower bounds, are multilinear. For instance, the permanent ($\Perm$), the determinant ($\Det$), the Nisan-Wigderson polynomials ($\NW$) and the iterated matrix multiplication polynomial ($\mathsf{IMM}$) are known to be hard for many natural classes of arithmetic circuits, homogeneous depth three circuits being one such class. Since for any $D\subseteq \F$ such that $|D| \geq 2$, $D^n$ is an interpolating set for multilinear polynomials, it seems natural to ask if there is a small homogeneous depth three arithmetic circuit which is functionally equivalent to any of these polynomials. 

Another reason why $\{0,1\}^n$ seems a natural domain to study functional algebraic computation is due to potential connections to boolean circuit lower bounds. It seems natural to ask if the techniques discovered in the quest for arithmetic circuit lower bounds can be adapted to say something interesting about questions in  boolean circuit complexity. And, \autoref{lem: acc to functional lb-intro} seems like an encouraging step in this direction.

\subsubsection{Functional lower bounds and partial derivatives }\label{sec: functional lb and partial derivatives}
Almost all the bounded depth arithmetic circuit lower bounds so far have been proved using techniques based on the partial derivatives of a polynomial. This includes exponential lower bounds for homogeneous depth-$3$ circuits \cite{nw1997} and  lower bounds for homogeneous depth-$4$ arithmetic circuits~\cite{gkks13, FLMS13, KLSS, KS14}. At a high level, the proofs have the following structure:
\begin{itemize}
\item Define a function $\Gamma : \F[\vecx] \rightarrow \N$, called the complexity measure, which serves as an indicator of the hardness of a polynomial. 
\item For all \emph{small} arithmetic circuits in the model of interest, show that $\Gamma$ has a non-trivial upper bound. 
\item For the target hard polynomial, show that $\Gamma$ is large. Comparing this with the upper bound in step $2$ leads to a contradiction if the hard polynomial had a small arithmetic circuit. 
\end{itemize}
The precise measure $\Gamma$ used in these proofs varies, but they all build upon the the notion of partial derivatives of a polynomial. 
The idea is to define $\Gamma(P)$ to be the dimension of a linear space of polynomials defined in terms of the partial derivatives of $P$. In the syntactic set up, if a circuit $C$ computes a polynomial $P$, then any partial derivative of $C$  must be equivalent to the corresponding partial derivative of $P$. This observation along with bounds on the dimension of the partial derivative based linear spaces, led to circuit lower bounds.

However, this clearly breaks down in the case when our only guarantee is that the circuit $C$ and the polynomial $P$ agree as functions on all of $\{0,1\}^n$. Apriori, it is not clear if we can say anything meaningful about how the partial derivatives of $C$ and those of $P$ are related to each other. An extreme case of this is the following example. Let the polynomials $P$ and $Q$ be defined as follows: 
\[
P = \left(\sum_{i = 1}^n x_i \right)^n
\]
and 
\[
Q = P \mod I_0
\]
Here $I_0$ is the ideal generated by the polynomials $\{x_i^2-x_i : i \in [n]\}$. The following items follow easily from the definitions:
\begin{itemize}
\item $\forall \vecx \in \{0,1\}^n, P(\vecx) = Q(\vecx)$.
\item The dimension of the span of partial derivatives of $P$ is at most $n$.
\item The dimension of the span of partial derivatives of $Q$ is at least $2^n$. This follows from the fact that the leading monomial of $Q$ is $x_1\cdot x_2 \cdots x_n$. 
\end{itemize}
So, clearly the dimension of the partial derivatives of two polynomials which are functionally the same over $\{0,1\}^n$ can be wildly different. Thus, it seems tricky to extend the proofs of syntactic lower bounds to the functional setup. Nevertheless, we do manage to get around this obstacle in certain cases as our results in the next section show. Moreover, we also show that a general solution to this question offers a possibility of proving new lower bounds for boolean circuits, that have so far been beyond our reach so far.

\subsection{Our results}
We now state our main results. 

As our first result, we show functional lower bounds for homogeneous\footnote{Our lower bounds require that the formal degree of the circuit and the degree of the polynomial are \emph{close} to each other. Homogeneity guarantees this condition, but is a much stronger condition than what we need for our proofs to work. } depth-$3$ circuits. In the syntactic setting such lower bounds were first shown by Nisan and Wigderson \cite{nw1997} using the partial derivative of a polynomial as the complexity measure. However, as we discussed in \autoref{sec: functional lb and partial derivatives}, partial derivative based proofs do not extend to the functional setting in a 
straightforward manner. We get around this obstacle by working with a different but related complexity measure. We now formally state the theorem : 
\begin{theorem}\label{thm: depth 3 lower bound}
Let $\F$ be any field. There exists a family $\{P_d\}$ of polynomials of degree $d$ in $n = \poly(d)$ variables in $\VNP$ such that any $\Sigma\Pi\Sigma$ circuit of formal degree $d$ which is functionally equivalent to $P_d$ over $\{0,1\}^n$ has size at least $\exp\left( \Omega\left(d\log n \right)\right)$.
\end{theorem}

As our second result, we show similar functional analogues of the homogeneous depth-$4$ lower bounds of \cite{KLSS,KS14} but under the restriction that the depth-$4$ circuit computes a polynomial of \emph{low individual degree}. As discussed in the introduction, such lower bounds for depth-$4$ circuits with bounded bottom fan-in but unbounded individual degree would imply that $\sp \neq \acc$, and would be a major progress on the question of boolean circuit lower bounds. 
\begin{theorem}\label{thm: depth 4 lower bound}
Let $\F$ be any field. There exists a family $\{P_d\}$ of polynomials of degree $d$ in $n = \poly(d)$ variables in $\VNP$ such that any $\Sigma\Pi\Sigma\Pi$ circuit of formal degree $d$ and individual degree $O(1)$ which is functionally equivalent to $P_d$ over $\{0,1\}^n$ has size at least $\exp\left( \Omega\left(\sqrt{d}\log n \right)\right)$.
\end{theorem}

Our techniques for the proof of \autoref{thm: depth 4 lower bound} are again different from the proofs of homogeneous depth-$4$ lower bounds in the syntactic setting. We introduce a family of new complexity measures, which are functional in their definition (as opposed to partial derivative based measures), and use them to capture functional computation. The family of measures, called \emph{Shifted Evaluation dimension} is a shifted analogue of the well known notion of evaluation dimension, which has had many applications in algebraic complexity (for instance, in multilinear formula, circuit lower bounds~\cite{raz2004, Raz06, raz-yehudayoff}). We believe that the measure is of independent interest, and could have other potential applications. 

\paragraph{Elementary symmetric polynomials : }
In their paper~\cite{nw1997}, Nisan and Wigderson showed an exponential lower bound  on the size of homogeneous depth-$3$ circuits computing the elementary symmetric polynomials. A curious consequence of our proof, is that we are unable to show an analogue of \autoref{thm: depth 3 lower bound} for the elementary symmetric polynomials. One of the reasons for this is the fact that the elementary symmetric polynomials have a \emph{small} evaluation dimension complexity (the complexity measure used for this lower bound), hence our proof technique fails. However, it turns out the at least over fields of sufficiently large characteristic, there are polynomial sized depth-$3$ circuits of low formal degree which are functionally equivalent to the elementary symmetric polynomials  over $\{0,1\}^n$. The upper bounds are based on the simple observation that for any $d$ and $x \in \{0,1\}^n$, the value of $Sym_d(x)$ (elementary symmetric polynomial of degree $d$) is equal to $\binom{h(x)}{d}$, where $h(x) = \sum_i x_i$ is the hamming weight of $x$. In particular, for $d = 1$, the polynomial $\sum_i x_i$ is functionally equivalent to $Sym_1$, the polynomial $\frac{(\sum_i x_i)(\sum_i x_i - 1)}{2}$ is functionally equivalent to $Sym_2$ and so on. In particular, there is a polynomial which is a product of $d$ affine forms which is equivalent to $Sym_d$. 
However, over fields of low characteristic, the complexity of the elementary symmetric polynomials for functional computation by depth-$3$ (or even depth-$4$) circuits is not clear to us and is an interesting open question. 

\paragraph{Comparison to Kayal, Saha, Tavenas~\cite{KST15} : } In a recent independent result, Kayal, Saha and Tavenas showed exponential lower bounds for depth-$4$ circuits of bounded individual degree computing an explicit polynomial in $\VP$. Their proof uses a complexity measure called \emph{skew shifted partials} which is very similar in spirit to the notion of \emph{shifted evaluation dimension}, the complexity measure we use. Even though the results seem related, none of them subsumes the other.  For our proof, we require that the formal degree of the depth-$4$ circuit is small (homogeneity), in addition to the individual degree being small, whereas in~\cite{KST15} the authors only require the individual degree of the circuit to be small. In this sense, their result is for a more general model than ours. However, for our lower bounds, we only require the circuit to agree with the target hard polynomial over $\{0,1\}^n$ while the proof in~\cite{KST15} is for syntactically computing the hard polynomial. 
Hence, the results are incomparable. 


\subsection{Organization of the paper}
We set up some notations to be used in the rest of the paper in \autoref{sec:notation}. We prove the connections between functional lower bounds for depth-$4$ circuits and  lower bounds for $\acc$ in \autoref{sec: functional lb and acc}. We introduce our main complexity measure in \autoref{sec:complexity measure}. We define and study the properties of the hard polynomials for our lower bounds in \autoref{sec:NW}. We present the proof of \autoref{thm: depth 3 lower bound} in \autoref{sec:depth 3} and the proof of \autoref{thm: depth 4 lower bound} in \autoref{sec:depth 4}.

\section{Notation}\label{sec:notation}
We now setup some notation to be used for the rest of the paper. 
\begin{itemize}

\item Throughout the paper, we shall use bold-face letters such as $\vecx$ to denote a set $\set{x_1,\dots, x_n}$. 
Most of the times, the size of this set would be clear from context. 
We shall also abuse this notation to use $\vecx^\vece$ to refer to the monomial $x_1^{e_1}\cdots x_n^{e_n}$. 

\item The set of formal variables in this paper denoted by $\vecx$ of size $n$ shall often be partitioned into sets $\vecy$ and $\vecz$. We shall use $n_y$ and $n_z$ to denote the sizes of $\vecy$ and $\vecz$ respectively. 

\item For an integer $m > 0$, we shall use $[m]$ to denote the set $\set{1,\dots, m}$. 

\item We shall use the short-hand $\partial_{\vecx^{\vece}}(P)$ to denote
\[
\frac{\partial^{e_1}}{\partial x_1^{e_1}}\inparen{ \frac{\partial^{e_2}}{\partial x_2^{e_2}}\inparen{\cdots \inparen{ P }\cdots}}.
\]

\item For a set of polynomials $\mathcal{P}$ shall use $\partial_{\vecy}^{=k}\mathcal{P}$ to denote the set of all $k$-th order partial derivatives of polynomials in $\mathcal{P}$ with respect to $y$ variables only, and $\partial_{\vecy}^{\leq k}\mathcal{P}$ similarly. 

Also, $\vecx^{=\ell} \mathcal{P}$ shall refer to the set of polynomials of the form $\vecx^{\vece} \cdot P$ where $\deg(\vecx^{\vece}) = \ell$ and $P \in \mathcal{P}$. Similarly $\vecx^{\leq \ell} \mathcal{P}$.  

\item For a polynomial $P \in \F[\vecx]$ and for a set $S \subseteq\F^n$, we shall denote by $\eval_S(P)$ the vector of the evaluation of $P$ on points in $S$ (in some natural predefined order like say the lexicographic order). 
For a set of vectors $V$, their span over $\F$ will be denoted by $\span(V)$ and their dimension by $\dim(V)$. 
\end{itemize}

\section{Functional lower bounds for depth-$4$ circuits and  $\acc$}\label{sec: functional lb and acc}
In this section, we show that strong enough functional lower bounds for even very special depth-$4$ arithmetic circuits are sufficient to imply new lower bounds for $\acc$. The proof follows from a simple application of a well known characterization of $\acc$ by Yao~\cite{yao85} and  Beigel and Tarui~\cite{beigeltarui94}. The following version of the theorem is from Arora-Barak~\cite{arorabarak}

\begin{theorem}[\cite{yao85, beigeltarui94}]\label{thm: acc to sym}
If a function $f:\{0,1\}^n \rightarrow \{0,1\}$ is in $\acc$, then $f$ can be computed by a depth $2$ circuit with a symmetric gate with quasipolynomial $\left(\exp(\log^{O(1)} n)\right)$ fan-in  at the output level and $\vee$ gates with polylogarithmic $\left(\log^{O(1)} n\right)$  fan-in at the bottom level. 
\end{theorem}

We now prove the following lemma which shows {\it functional} upper bound for $\acc$. 

\begin{lemma}\label{lem: acc to functional ub}
Let $\F$ be any field of characteristic zero or at least $\exp\left(\omega\left(\poly(\log n)\right)\right)$. If a function $f:\{0,1\}^n \rightarrow \{0,1\}$ is in $\acc$, then there exists a polynomial $P_f \in \F[x_1, x_2, \ldots, x_n]$ such that the following are true: 
\begin{itemize}
\item For every $\vecx \in \{0,1\}^n$, $f(\vecx) = P_f(\vecx)$. 
\item $P_f$ can be computed by a quasipolynomial sized $\Sigma\!\wedge\!\Sigma\Pi$ circuit with bottom fan-in at most $\poly(\log n)$, which are depth-$4$ circuits where the product gates in the second level are powering gates. 
\end{itemize}
\end{lemma}
\begin{proof}
From \autoref{thm: acc to sym}, we know that there exists a symmetric function $h$ and multilinear polynomials $g_1, g_2, \ldots, g_t$  such that 
\begin{itemize}
\item $t = \exp(\poly(\log n))$. 
\item For every $\vecx \in \{0,1\}^n$, $f(\vecx) = h(g_1(\vecx), g_2(\vecx), \ldots, g_t(\vecx))$. 
\item Each $g_i$ is a multilinear polynomial in at most $\poly(\log n)$ variables. 
\item For every $\vecx \in \{0,1\}^n$ and $j \in [t]$, $g_j(\vecx) \in \{0,1\}$. 
\end{itemize}
From the last item above, we know that the $g_i$s only take boolean values on inputs from $\{0,1\}^n$. Since $h$ is symmetric, it follows that its value on boolean inputs only depends upon the hamming weight of its input. Hence, $h$ is in fact a function of $\sum_{i \in [t]} g_i$. Therefore, over any field of characteristic zero or larger than $t$, there exists a univariate polynomial $P_h$ of degree at most $t$ over reals, such that 
\[
\forall \vecx \in \{0,1\}^n, h\left(g_1(\vecx), g_2(\vecx), \ldots, g_t(\vecx)\right) = P_h\left(\sum_{i \in [t]} g_i(\vecx)\right)
\] 
The lemma now follows from the fact that each $g_i$ is a multilinear polynomial in $\poly(\log n)$ variables. 
\end{proof}

\autoref{lem: acc to functional ub} now immediately implies the following lemma. 
\begin{lemma}~\label{lem: acc to functional lb}
Let $\F$ be any field of characteristic zero or at least $\exp\left(\omega\left(\poly(\log n)\right)\right)$. Then, an \\
$\exp\left(\omega\left(\poly(\log n)\right)\right)$ functional lower bound for a function on $n$ variables for $\Sigma\wedge\Sigma\Pi^{[\poly(\log n)]}$ circuits over $\F$ would imply that $f$ is not in  $\acc$. 
\end{lemma}

\section{The complexity measure}~\label{sec:complexity measure}
In the lower bounds for homogeneous depth four circuits \cite{KLSS,KS14}, the complexity measure used was the \emph{dimension of projected shifted partial derivatives}. The following definition is not the same as used in \cite{KLSS,KS14}, but this slight variant would be easier to work with for our applications. We abuse notation to call it ``projected shifted partial derivatives'' as it continues to have the essence of the original definition. A discussion on the precise differences between the following definition and the original definition of \cite{KLSS,KS14} is present in \autoref{sec:pspd-discussion}

\begin{definition}[Projected shifted partial derivatives]\label{defn:pspd}
Let $\vecx = \vecy \sqcup \vecz$ with $|\vecy| = n_y$ and $|\vecz| = n_z$, and let $S$ be the set of all strings in $\set{0,1}^{n_y + n_z}$ that are zero on the first $n_y$ coordinates. If $k, \ell$ are some parameters, the \emph{dimension of projected shifted partial derivatives} for any polynomial $P(\vecy, \vecz) \in \F[\vecy, \vecz]$, denoted by $\Gamma_{k,\ell}^{\mathrm{PSPD}}(P)$, is defined as
\[
\Gamma_{k,\ell}^{\mathrm{PSPD}}(P) \spaced{:=} \dim\inbrace{\eval_{S}\inparen{\vecz^{=\ell} \partial_{\vecy}^{=k}(P)}}.\qedhere
\]
\end{definition}

The above measure is still syntactic as partial derivatives are not useful in the functional setting. For the functional setting, we shall use a different measure for our lower bound that we call \emph{the shifted evaluation dimension}.  We now define the complexity measure that we shall be using to prove the lower bound. For brevity, we shall assume that our set of variables $\vecx$ is partitioned into $\vecy$ and $\vecz$.  For our proofs, we shall use a carefully chosen partition. We now formally define the notion of \emph{shifted evaluation dimension} of a polynomial below. 


\begin{definition}[Shifted evaluation dimension]
Let $\ell$ and $k$ be some parameters and let $\vecx = \vecy \sqcup \vecz$ such that $|\vecy| = n_y$ and $|\vecz| = n_z$.  For any polynomial $P \in \F[\vecy, \vecz]$, define $\Gamma_{k, \ell}(P)$ as
\[
\Gamma_{k,\ell}^{\mathrm{SED}}(P) \spaced{:=} \dim\inbrace{\eval_{\{0,1\}^{n_z}}\inparen{\vecz^{=\ell} \cdot \{P(\veca,\vecz) : \veca \in \{0,1\}^{n_y}_{\leq k} \}}}.\qedhere
\]
\end{definition}

Informally, for every polynomial $P$, we fix a partition of the input variables into $\vecy$ and $\vecz$ and generate a linear space by the following algorithm. 
\begin{itemize}
\itemsep 0pt
\item We take the projections of $P$ obtained by setting each of the $y$ variables to $0,1$ such that the number of $y$ variables set to $1$ is at most $k$.
\item  We shift the polynomials obtained in step $1$ by all monomials in variables $\vecz$ of degree  $\ell$.
\item Observe that the polynomials obtained at the end of step two are polynomials  only in the $\vecz$ variables. We now look at the evaluation vectors of these polynomials over $\{0,1\}^{n_z}$. 
\end{itemize}
The complexity measure of the polynomial $P$ is defined as the dimension of the linear space generated by the vectors obtained at the end of step $3$ in the algorithm above. For our proof, we will pick a careful partition of the variables $\vecx$ into $\vecy$ and $\vecz$ and look at $\Gamma_{k, \ell}^{\mathrm{SED}}(P)$. The following lemma highlights the key reason of utility of the above measure to functional lower bounds. 

\begin{lemma}[Functional equivalence and shifted evaluation dimension]\label{lem: complexity measure utility}
Let $P \in \F[\vecx]$ and $Q \in \F[\vecx]$ be any two polynomials which are functionally equivalent over $\{0,1\}^n$. Then, for every choice of $k$, $\ell$ and partition $\vecx = \vecy \sqcup \vecz$
\[
\Gamma_{k, \ell}^{\mathrm{SED}}(P) \spaced{=} \Gamma_{k, \ell}^{\mathrm{SED}}(Q)
\]
\end{lemma}
\begin{proof}
The proof easily follows from the fact that the measure $\Gamma_{k, \ell}^{\mathrm{SED}}(P)$ is the dimension of a linear space which is generated by vectors which correspond to evaluations of $P$ over subcubes of $\{0,1\}^n$. Hence, it would be the same for any two polynomials which agree as functions over $\{0,1\}^n$.  
\end{proof}

\begin{remark*}
Observe that a lemma analogous to \autoref{lem: complexity measure utility} is not true in general for partial derivative based measures. And hence, the proofs for syntactic lower bounds which are based on such measures does not immediately carry over to the functional setting. 
\end{remark*}

\subsection{Evaluations  vs  partial derivatives}
In this section, we show that  for polynomials of low individual degree, the notion of shifted evaluation dimension can be used as a proxy for the notion of shifted partial derivatives. This is the key observation that drives the proofs of \autoref{thm: depth 3 lower bound} and \autoref{thm: depth 4 lower bound}. We first consider the case when the polynomial is \emph{set-multilinear} in which case derivatives can be directly related to careful evaluations. 

\subsubsection{For set-multilinear polynomials}

The explicit polynomials we shall be working with in this paper would be \emph{set-multilinear}. An example to keep in mind is $\Det_n$ or $\Perm_n$ where the variables can be partitioned into rows and each monomial involves exactly one variable from each part. 

\begin{definition}[Set-multilinear polynomials] A polynomial $P$ is said to be \emph{set-multilinear} with respect to the a partition $\vecx = \vecx_1 \sqcup \cdots \sqcup \vecx_r$ if every monomial of $P$ involves exactly\footnote{sometimes in the literature the word `exactly' is replaced by `at most' but in this paper we would be dealing with this definition.} one variable from each $\vecx_i$. 
\end{definition}

We begin with the following simple observation. 

\begin{observation}\label{lem: multilinear equivalence}
Let $P \in \F[\vecx]$ be a set-multilinear with respect to a partition $\vecx = \vecx_1 \sqcup \cdots \sqcup \vecx_r$. Let $\vecy = \vecx_1 \union \cdots \union \vecx_k$ for some $k \leq r$ and let $\vecz = \vecx \setminus \vecy$. Then, for any degree $k$ monomial $\vecy^\vece$ that is set-multilinear with respect to $\vecx_1 \sqcup \cdots \sqcup \vecx_k$, we have
\[
\frac{\partial P}{\partial \vecy^{\vece}} \spaced{=} P(\vece, \vecz).
\]  
\end{observation}
\begin{proof}
We shall prove this by induction on $k$. Suppose $\vecy = \vecx_1$ and $y_1 \in \vecx_1$. Since $P$ is set-multilinear, we can write $P$ as
\[
P(\vecx_1,\cdots, \vecx_r) \spaced{=} \sum_{y_i \in \vecx_1} y_i \cdot P_i(\vecx_2,\cdots, \vecx_r).
\]
Hence it follows that $\partial_{y_1}(P)$ equals $P_1$, which is also the partial evaluation of $P$ where $y_1$ is set to $1$ and all other $y_i \in \vecx_1$ is set to zero. Hence, if $y_1 = \vecy^\vece$, then $\partial_{y_1}(P) = P(\vece, \vecx_2,\cdots, \vecx_r)$. 
The claim follows by repeating this argument on $P(\vece,\vecx_2,\cdots, \vecx_r)$ which continues to be set-multilinear.
\end{proof}

\autoref{lem: multilinear equivalence} immediately implies the following corollary, which shows that for set-multilinear polynomials {\emph shifted evaluation dimension} and {\emph shifted partial derivatives} are the same quantity if we choose our set of derivatives carefully. 

\begin{corollary}\label{cor: taylor mult}
Let $P(\vecx)$ be a set-multilinear polynomial with respect to $\vecx = \vecx_1 \sqcup \cdots \sqcup \vecx_r$. Suppose $\vecy = \vecx_1 \union \cdots \union \vecx_k$ and $\vecz = \vecx \setminus \vecy$. Then if we consider the dimension of projected shifted partials with respect to set-multilinear monomials in $\vecy$, we have
\[
\Gamma_{k, \ell}^{\mathrm{PSPD}}(P) \spaced{\leq} \Gamma_{k, \ell}^{\mathrm{SED}}(P).
\]  
\end{corollary}

\subsubsection{For low individual degree polynomials}

We now proceed to show that an \emph{approximation} of the \autoref{cor: taylor mult} also holds for polynomials of low individual degree. 
\begin{lemma}\label{lemma: low ind degree taylor}
Let $P(\vecy, \vecz)$ be a polynomial with individual degree at most $r$. Then, for every choice of parameters $k$ and $\ell$ 
\[
\set{ P(\veca, \vecz) : a \in \{0,1\}^{n_y}_{\leq k} } \spaced{\subseteq} \span\inparen{\inparen{\partial^{\leq rk} P}_{\vecy = \mathbf{0}}}.
\]  
\end{lemma}
\begin{proof}
For the rest of this proof, we shall think of $P$ as an element $P_\vecz(\vecy) \in \F[\vecz][\vecy]$. Let $\veca$ be any point in $\{0,1\}^{n_y}$. Then by the Taylor's expansion, we know that 
\[
P_\vecz(\vecy + \veca) \spaced{=} \sum_{\vece} \veca^{\vece} \cdot \partial_{\vecy^{\vece}}(P_{\vecz})(\vecy)
\]
If the support of $\veca$ is at most $k$, then for every $\vece$ such that $\|\vece \|_0 > k$, we would have $\veca^{\vece} = 0$. Moreover, since $P$ is a polynomial of individual degree at most $r$, it follows that if any coordinate of $\vece$ is more than $r$ then
\[
\partial_{\vecy^{\vece}}(P_{\vecz})  = 0.
\]
In summary, for any $\veca$ such that $\|\veca\|_0 \leq k$, 
\begin{eqnarray*}
P_\vecz(\vecy + \veca) & = & \sum_{\substack{\vece : \|\vece\|_0 \leq k,\\ \|\vece\|_1 \leq rk}}\veca^{\vece} \cdot \partial_{\vecy^{\vece}}(P_{\vecz})(\vecy)\\
\implies P_\vecz(\veca) \quad = \quad  P(\veca, \vecz) & = & \sum_{\substack{\vece : \|\vece\|_0 \leq k,\\ \|\vece\|_1 \leq rk}}\veca^{\vece} \cdot \inparen{\partial_{\vecy^{\vece}}(P_{\vecz})}_{\vecy = \mathbf{0}} \quad \in \quad \span\inparen{\inparen{\partial^{\leq rk} P}_{\vecy = \mathbf{0}}}.\qedhere
\end{eqnarray*}
\end{proof}
We are now ready to prove our main technical claim of this section. 
\begin{lemma}\label{lemma: eval dim vs partial derivatives}
Let $P(\vecy, \vecz)$ be a polynomial with individual degree at most $r$. Then, for every choice of parameters $k$ and $\ell$, 
\[
\Gamma_{k, \ell}^{\mathrm{SED}}(P) \spaced{\leq} \Gamma_{rk, \ell}^{\mathrm{PSPD}}(P)
\]  
\end{lemma}
\begin{proof}
From \autoref{lemma: low ind degree taylor}, we know that 
\begin{eqnarray*}
\set{P(\veca, \vecz) : \veca \in \{0,1\}^{n_y}_{\leq k}} &\subseteq& \span\inparen{\inparen{\partial^{\leq rk} P}_{\vecy = \mathbf{0}}} \\
\implies \set{\vecz^{=\ell} \cdot P(\veca, \vecz) : \veca \in \{0,1\}^{n_y}_{\leq k}} &\subseteq& \span\inparen{\vecz^{=\ell} \cdot \inparen{\partial^{\leq rk} P}_{\vecy = \mathbf{0}}}
\end{eqnarray*}
By looking at the evaluation vectors on $\set{0,1}^{n_z}$, 
\begin{eqnarray*}
\left\{ \eval_{\set{0,1}^{n_z}}\left(\vecz^{=\ell} \cdot  P(\veca, \vecz) \right): \veca \in \{0,1\}^{n_y}_{\leq k} \right\} &\subseteq& \span\left( \eval_{\set{0,1}^{n_z}}\left(\vecz^{=\ell}\cdot  \inparen{\partial^{\leq rk} P}_{\vecy = \mathbf{0}} \right)\right)\\
 & = & \span\left( \eval_{\set{0}^{n_y} \times \set{0,1}^{n_z}}\left(\vecz^{=\ell}\cdot \partial^{\leq rk} P \right)\right)
\end{eqnarray*}
Taking the dimension of the linear spans on both sides completes the proof. 
\end{proof}

\section{Nisan-Wigderson polynomial families}~\label{sec:NW}
In this section, we formally define the family of Nisan-Wigderson polynomials and mention some known results about lower bounds on the their projected shifted partials complexity~\cite{KLSS, KS14, KumarSaptharishi15}. These bounds will be critically used in our proof. 
\begin{definition}[Nisan-Wigderson polynomial families]~\label{def:NW final}
Let $d,m,e$ be arbitrary parameters with $m$ being a power of a prime, and $d,e\leq m$. 
Since $m$ is a power of a prime, let us identify the set $[m]$ with the field $\F_m$ of $m$ elements. 
Note that since $d \leq m$, we have that $[d] \subseteq \F_m$. 
The Nisan-Wigderson polynomial with parameters $d,m,e$, denoted by $\NW_{d,m,e}$ is defined as
\[
\NW_{d,m,e}(\vecx) \spaced{=} \sum_{\substack{p(t) \in \F_m[t]\\ \deg(p) < e}} x_{1,p(1)}\dots x_{d, p(d)}
\]
That is, for every univariate polynomial $p(t) \in \F_m[t]$ of degree less that $e$, we add one monomial that encodes the `graph' of $p$ on the points $[d]$. 

This is a homogeneous, multilinear polynomial of degree $d$ over $dm$ variables with exactly $m^e$ monomials. Furthermore, the polynomial is \emph{set-multilinear} with respect to $\vecx = \vecx_1 \sqcup \cdots \sqcup \vecx_d$ where $\vecx_i = \set{x_{i1},\cdots, x_{im}}$. 
\end{definition}

We now state the following lemma which shows a lower bound on the $\Gamma_{k, \ell}^{\mathrm{PSPD}}(\NW_{d,m,e})$ for an appropriate choice of parameters. We will then use this bound along with \autoref{cor: taylor mult} to show a lower bound on $\Gamma_{k, \ell}^{\mathrm{SED}}(NW_{d,m,e})$. The lower bound on $\Gamma_{k, \ell}^{\mathrm{PSPD}}(\NW_{d,m,e})$ was shown in two independent proofs by Kayal et al.~\cite{KLSS} and by Kumar and Saraf~\cite{KS14}. The version stated below is from a strengthening of these bounds by Kumar and Saptharishi~\cite{KumarSaptharishi15}. 
 
\begin{lemma}\label{lem: psd lower bound for nw} 
For every $d$ and $k = O(\sqrt{d})$ there exists parameters $m,e,  \epsilon$ such that $m = \Theta(d^2)$
 and $\epsilon = \Theta\pfrac{\log d}{\sqrt{d}}$ with
\begin{eqnarray*}
  m^k & \geq & (1+\epsilon)^{2(d-k)}\\
  m^{e-k} & = & \pfrac{2}{1+\epsilon}^{d-k} \cdot \poly(m). 
\end{eqnarray*}
For such a choice of parameters, let $\vecx = \setdef{x_{ij}}{i\in [d]\;,\; j\in [m]} = \vecx_1 \sqcup \cdots \sqcup \vecx_d$ where $\vecx_i = \set{x_{i1}, \ldots, x_{im}}$. Let $\vecy = \vecx_1 \sqcup \cdots \sqcup \vecx_k$ and $\vecz = \vecx \setminus \vecy$. If $\ell$ is a parameter that satisfies $\ell = \frac{n_z}{2} (1 - \epsilon)$, then over any field $\F$, we have\footnote{We remark that in the calculations in \cite{KLSS, KS14, KumarSaptharishi15}, the shifted monomials consist of both the $\vecy$ and $\vecz$ variables, while here we only shift by $\vecz$ variables. But the calculations still go through since the parameters continue to satisfy the constraints needed for soundness of the calculation. }
\[
\Gamma_{k,\ell}^{\mathrm{PSPD}}(\NW_{d,m,e}(\vecy, \vecz)) \spaced{\geq} \binom{n_z}{\ell + d - k} \cdot \exp(-O(\log^2 d)).
\]
\end{lemma}

From \autoref{cor: taylor mult}, we immediately have the following crucial lemma. 
\begin{lemma}\label{lem: NW complexity lower bound}
Let $d, m, e, \ell$ be parameters as defined in \autoref{lem: psd lower bound for nw} and let $\vecy$ and $\vecz$ be the partition of variables $\vecx$ as in \autoref{lem: psd lower bound for nw}. 
Then,over any field $\F$, we have
\[
\Gamma_{k,\ell}^{\mathrm{SED}}(\NW_{d,m,e}(\vecy, \vecz)) \spaced{\geq} \binom{n_z}{\ell + d - k} \cdot \exp(-O(\log^2 d)).
\]
\end{lemma}

\section{Functional lower bounds for depth-$3$ circuits}~\label{sec:depth 3}
In this section, we complete the  proof of \autoref{thm: depth 3 lower bound}. We start by defining the exact hard polynomial for which our lower bound is shown. 
\subsection*{Hard polynomials for the lower bound}
We will prove \autoref{thm: depth 3 lower bound} for the polynomial $\NW_{d,m,e}$ for an appropriate choice of the parameters.
\begin{lemma}\label{lem : nw partial derivative lower bounds}
Let the parameters $e$ and $d$ be chosen so that $e = d/2-1$, and let $k = e+1$. Let the variables $\vecx$ in $\NW_{d,m,e}$ be partitioned into $\vecy = \setdef{x_{ij}}{i\in [k], j\in [m]}$ and $\vecz = \vecx \setminus \vecy$. 
Then 
\[
\Gamma_{k, 0}^{\mathrm{SED}}(\NW_{d,m,e}(\vecy, \vecz)) \spaced{\geq} m^{d/2}.
\]
\end{lemma}
\begin{proof}
Let the set of monomials $S$ be defined as  
\[
S = \left\{\prod_{i = 1}^k x_{i,j_i} : j_i \in [m] \right\}
\]
Observe that for every monomial $\vecx^\alpha$ in $S$, the partial derivative of $\NW_{d,m,e}$ with respect to $\vecx^\alpha$, is a monomial in $\vecz$. This is due to the fact that $e < d/2$ and no two distinct univariate polynomials of degree $d/2$ can agree at more than $d/2$ many points. Moreover for  every two distinct monomials $\vecx^\alpha$ and $\vecx^\beta$ in $S$, 
\[
\frac{{\partial \NW_{d,m,e}}}{{\partial 
\vecx^\alpha}} \spaced{\neq} \frac{\partial \NW_{d,m,e}}{\partial \vecx^\beta}
\]
Hence, 
\[
\Gamma_{k, 0}^{\mathrm{PSPD}}(\NW_{d,m,e}) = |S| = m^{d/2} 
\]
Since $\NW_{d,m,e}$ is a set-multilinear with respect to the rows of variable matrix, by \autoref{lem: multilinear equivalence}, it follows that 
\[
\Gamma_{k, 0}^{\mathrm{SED}}(\NW_{d,m,e}) = m^{d/2} \qedhere
\]
\end{proof}

\subsection*{Complexity of the model}
\begin{lemma}\label{lem: depth 3 circuit complexity upper bound}
The $C(\vecx)$ be a $\SPS$ circuit of formal degree $d$ and top fan-in $s$. Then, for all choices of $k$ and any partition of $\vecx$ into $\vecy$ and $\vecz$,  
\[
\Gamma_{k, 0}^{\mathrm{SED}}(C) \spaced{\leq} s\cdot 2^d 
\]
\end{lemma}
\begin{proof}
Observe that for any choice of $k$ and $\ell$, $\Gamma_{k,\ell}^{\mathrm{SED}}$ is a subadditive measure. Therefore, it is enough to upper bound the value of $\Gamma_{k, 0}^{\mathrm{SED}}()$ for every product gate in $C$ by $2^d$. 
Let 
\[
Q(\vecy, \vecz) = \prod_{i = 1}^d L_i
\] 
be any product gate of formal degree at most $d$ in $C$. Since each $L_i$ is a linear form, we can express it as $L_i = L_{yi} + L_{zi}$, where $L_{yi}$ and $L_{zi}$ are the parts of $L_i$ consisting entirely of $\vecy$ and $\vecz$ variables respectively. Therefore, 
\[
Q(\vecy, \vecz) = \sum_{S\subseteq [d]}\prod_{i\in S} L_{yi} \cdot \prod_{j \notin S} L_{zj}
\] 
Now observe that by 
\[
\left\{  Q(\veca, \vecz) : \veca \in \{0,1\}^{n_y} \right\} \spaced{\subseteq} \span\left(\left\{\prod_{j \notin S} L_{zj} : S \subseteq [d]\right\} \right) 
\]
Therefore, 
\[
\Gamma_{k, 0}^{\mathrm{SED}}(C) \spaced{\leq}  2^d 
\]
The lemma now follows by subadditivity.
\end{proof}

\subsection*{Wrapping up the proof}
We are now ready to complete the proof of \autoref{thm: depth 3 lower bound}. 

\begin{theorem}
Let $\F$ be any field, and let $d,m,e$ be parameters such that $e = d/2-1$ and $m = \poly(d)$. Let $C$ be a $\SPS$ circuit of formal degree $d$ which is functionally equivalent to the polynomial $\NW_{d,m,e}$. Then
\[
 \text{Size}(C) \geq m^{d/2}/2^d
 \]
\end{theorem}
\begin{proof}
Let $k=e+1$ and consider a partition of variables into $\vecy$ and $\vecz$ where all the variables in the first $k$ rows of the variable matrix are labelled $\vecy$ and the remaining variables are labelled $\vecz$. Now, the theorem immediately follows from \autoref{lem : nw partial derivative lower bounds} and \autoref{lem: depth 3 circuit complexity upper bound}.
\end{proof}

\section{Functional lower bounds for depth-$4$ circuits}~\label{sec:depth 4}

In this section, we prove \autoref{thm: depth 4 lower bound}. We first define the family of polynomials for which our lower bounds apply. 

\subsection*{Hard polynomials for the lower bound}

For the proof of \autoref{thm: depth 4 lower bound}, we would have to show that a statement in the spirit of \autoref{lem: NW complexity lower bound} is also true for a \emph{random projection} of our hard polynomial. 
Even though we believe\footnote{In fact, \cite{KLSS, KS14} showed such statements to be true.} that this is true for the polynomial defined in \autoref{def:NW final}, for simplicity, we modify our hard polynomial and in turn prove a lower bound for the following variant of it. 

\begin{definition}[Hard polynomials for the lower bound]
Let $d, m, e$ be parameters as defined in \autoref{def:NW final}.  Let $p = p(m, d)$ be a parameter and let 
\[
t = \frac{dm}{p}
\]
The polynomial $\NWL$ is defined as 
\[
\NWL_{d,m,e,p} = \NW_{d, m, e}\left(L(x_{1,1}), L(x_{1,2}), \dots, L(x_{d,m}) \right)
\]
where for each $i \in [d], j \in [m]$, $L(x_{i,j})$ is defined as
\[
L(x_{i,j}) = \sum_{u = 1}^t x_{i,j,u}
\]
\end{definition}

For the rest of this proof, we set $p = (md)^{-0.1}$, and for brevity, we will indicate $\NWL_{d,m,e, (md)^{0.1}}$ by $\NWL_{d,m,e}$. Observe that setting $p$ sets $t$ to be equal to $(md)^{1.1}$. We conclude this section with the next lemma where we show that $\NWL_{d,m,e}$ is \emph{robust} under random restrictions where every variable is kept alive with a probability $p$.

\begin{lemma}\label{lem: robustness under random restrictions}
Let $p$ and $t$ be as stated above and let $n = dm$. Let $P$ be a random projection of $\NWL$ obtained by setting every variable in $\{x_{i,j,h} : i \in [d], j \in [m], h \in [t]\}$ to zero with a probability equal to $1-p$. Then, with a probability at least $1-o(1)$, $\NW_{d,m,e}$ is a projection of $P$. 
\end{lemma} 
\begin{proof}
For every $i \in [d]$, $j \in [m]$, define the set $A_{i,j}$ as
\[
A_{ij} = \{x_{i,j,h} : h \in [t]\}
\]
When every variable is being set to zero with a probability $1-p$, the probability that there exists an $i \in [d]$ and $j \in [m]$ such that all the variables in the set $A_{i,j}$ are set to zero is at most $dm(1-p)^t$. For $p = n^{-0.1}$, the probability is at most $n\dot (1-n^{-0.1})^{n^{1.1}}$ which is $\exp(-\Omega(n))$. 

Therefore, with a probability at least $1-\exp(-\Omega(n))$, each of the set $A_{i,j}$ has at least one variable alive in $P$. Now, we set all but one of them to zero for each $i,j$. Observe that the resulting projection of $P$ is precisely $\NW_{d,m,e}$ up to a relabelling of variables. This proves the lemma. 
\end{proof} 

It should be noted that the polynomial $\NWL$ continues to remain set-multilinear with respect to he rows of the variable matrix. 

\subsection*{Upper bound on the complexity of the model}
We now show the upper bound on $\Gamma_{k, \ell}^{\mathrm{SED}}(C)$ when $C$ is a depth-$4$ circuit of individual degree at most $r$ and bottom support $s$. We will use the following upper bound on $\Gamma_{k, \ell}^{\mathrm{PSPD}}(C)$ from \cite{KLSS, KS14}. 
\begin{lemma}\label{lem: depth 4 circuits psd}
Let $C(\vecy, \vecz)$ be a depth-$4$ circuit, of formal degree at most $d$ and bottom support at most $s$.  Let $k$ and $\ell$ be  parameters satisfying$\ell + ks < n_z/2$. Then
\[
\Gamma_{k,\ell}^{\mathrm{PSPD}}(C) \spaced{\leq} \text{Size}(C) \cdot \binom{O\left( \frac{d}{s} \right) + k }{k} \cdot \binom{n_z}{\ell + ks}\cdot \poly(n).
\]
\end{lemma}
The following lemma now immediately follows from \autoref{lem: depth 4 circuits psd} and \autoref{lemma: eval dim vs partial derivatives}. 
\begin{lemma}\label{lem: depth 4 circuits sed-low bottom support}
Let $C(\vecy, \vecz)$ be a depth-$4$ circuit, of formal degree at most $d$, individual degree at most $r$ and bottom support at most $s$.  Let $k$ and $\ell$ be  parameters satisfying $\ell + krs < n_z/2$. Then
\[
\Gamma_{k,\ell}^{\mathrm{SED}}(C) \spaced{\leq} \text{Size}(C) \cdot \binom{O\left( \frac{d}{s} \right) + kr }{kr} \cdot \binom{n_z}{\ell + krs}\cdot \poly(n_z).
\]
\end{lemma}

\subsection*{Wrapping up the proof}
\begin{theorem}
Let $d, m, e$ be parameters as defined in \autoref{lem: psd lower bound for nw}. Let $C$ be a $\SPSP$ circuit $C$ of formal degree $d$ and individual degree at most $r = O(1)$ over any field $\F$ such that $C$ is functionally equivalent to $\NWL_{d,m,e}$. Then, 
\[
\text{Size}(C) \geq \exp\left(\Omega\left(\sqrt{d}\log dm \right) \right)
\]
\end{theorem}
\begin{proof}
If the size of $C$ is larger than $\exp\left(\frac{\sqrt{d}\log dm}{1000r}  \right)$, then we are already done, else the size of $C$ is at most $\exp\left(\frac{\sqrt{d}\log dm}{1000r}  \right)$. 
Let us set every variable in $C$ and $\NWL_{d,m,e}$ to zero independently with a probability $1-(md)^{-0.1}$. The following claim easily follows via a standard application of the union bound. 
\begin{claim}
With probability at least $1-o(1)$ over the random restrictions as defined above, every product gate at the bottom level of $C$ with support at least $\frac{\sqrt{d}}{100r}$ is set to zero. 
\end{claim}
From the above claim and from \autoref{lem: robustness under random restrictions}, it follows that there is a $\SPSP$ circuit $C'$ of formal degree $d$ over $\F$ which is functionally equivalent to $\NW_{d,m,e}$. Let us relabel the variables as $\vecy$ and $\vecz$ as described in \autoref{lem: psd lower bound for nw}. Let $k = \sqrt{d}$ and let $\ell = \frac{n_z}{2}\cdot (1-\epsilon)$ where $\epsilon = O\left( \frac{\log d}{\sqrt{d}}\right)$ to be chosen shortly. By \autoref{lem: NW complexity lower bound}, we know that for this choice of  $k$ and $\ell$ 
\begin{eqnarray*}
\Gamma_{k,\ell}^{\mathrm{SED}}(\NW_{d,m,e}(\vecy, \vecz)) &\spaced{\geq}& \binom{n_z}{\ell + d - k} \cdot \exp(-O(\log^2 d))\\
&\spaced{\geq}& \binom{n_z}{\ell} \cdot (1+\epsilon)^{2d-2k} \cdot \exp(-O(\log^2 d))
\end{eqnarray*}
Moreover, by \autoref{lem: depth 4 circuits sed-low bottom support}, we know that
\begin{eqnarray*}
\Gamma_{k,\ell}^{\mathrm{SED}}(C') &\spaced{\leq}&  (dm)^{\sqrt{d}/1000r} \cdot \binom{O\left( \frac{\sqrt{d}}{r} \right) + kr }{kr} \cdot \binom{n_z}{{\ell + k\cdot r\cdot \frac{\sqrt{d}}{100r}}}\cdot \poly(n_z)\\
&\spaced{\leq}& (dm)^{\sqrt{d}/1000r} \cdot 2^{O(\sqrt{d})} \cdot \binom{n_z}{\ell} \cdot (1+\epsilon)^{\frac{d}{50}} \cdot \exp(O(\log^2 d))\\
&\spaced{\leq}& \exp{\left({\sqrt{d}\log d/100r}\right)} \cdot 2^{O(\sqrt{d})} \cdot \binom{n_z}{\ell} \cdot (1+\epsilon)^{\frac{d}{50}} \cdot \exp(O(\log^2 d))
\end{eqnarray*}
Now, observe that there exists a constant $c$ such that if $\epsilon$ is set to $\frac{c\log d}{\sqrt{d}}$, then 
\[
\Gamma_{k,\ell}^{\mathrm{SED}}(\NW_{d,m,e}) > \Gamma_{k,\ell}^{\mathrm{SED}}(C')
\]
But this is a contradiction since $C'$ computes $\NW_{d,m,e}$. This completes the proof. 
\end{proof}

\section{Open problems} 
 We end with some open questions : 
 \begin{itemize}
 \item The main challenge would be to improve \autoref{thm: depth 4 lower bound}, and prove it for the model of sums of powers of low degree polynomials. It is not clear to us if the complexity measure used in this paper would be useful.
 \item The functional lower bounds proved in this paper are for \emph{exact} functional computation. We believe that some of these bounds should also hold in the average case, where the circuit and the polynomial agree on a random point on $\{0,1\}^n$ with a high probability.
 It is not clear to us if the proof techniques in this paper can be adapted to say something in the average case setting. The most natural attempt to generalize the proofs seem to hit a \emph{matrix rigidity} like obstacle.  
 \end{itemize}
\section*{Acknowledgement} 
Part of this work was done while the third author was visiting Rutgers. We are grateful to Eric Allender and DIMACS for funding the visit. We are also grateful to Pravesh Kothari and Madhu Sudan for many helpful conversations.

\bibliographystyle{customurlbst/alphaurlpp}
\bibliography{references}

\begin{thebibliography}{FSTW15}

\bibitem[AB09]{arorabarak}
Sanjeev Arora and Boaz Barak.
\newblock {\em Computational Complexity: A Modern Approach}.
\newblock Cambridge University Press, New York, NY, USA, 1st edition, 2009.

\bibitem[BT94]{beigeltarui94}
Richard Beigel and Jun Tarui.
\newblock \href {http://dx.doi.org/10.1007/BF01263423} {On ACC}.
\newblock {\em Computational Complexity}, 4(4):350--366, 1994.

\bibitem[FLMS14]{FLMS13}
Herv{\'{e}} Fournier, Nutan Limaye, Guillaume Malod, and Srikanth Srinivasan.
\newblock \href {http://doi.acm.org/10.1145/2591796.2591824} {Lower bounds for
  depth 4 formulas computing iterated matrix multiplication}.
\newblock In {\em \STOC{2014}}, pages 128--135, 2014.
\newblock Pre-print available at \parseECCC{TR13/100}.

\bibitem[FSTW15]{FSTW15}
Michael~A. Forbes, Amir Shpilka, Iddo Tzameret, and Avi Wigderson.
\newblock Proof Complexity Lower Bounds from Algebraic Circuit Complexity.
\newblock Manuscript, 2015.

\bibitem[GKKS14]{gkks13}
Ankit Gupta, Pritish Kamath, Neeraj Kayal, and Ramprasad Saptharishi.
\newblock \href {http://dx.doi.org/10.1145/2629541} {Approaching the Chasm at
  Depth Four}.
\newblock {\em Journal of the ACM}, 61(6):33:1--33:16, 2014.
\newblock \pCCC{2013}.
\newblock Pre-print available at \parseECCC{TR12/098}.

\bibitem[GP14]{GrochowPitassi14}
\mfbiberr{toupdate(GrochowPitassi14): journal}Joshua~A. Grochow and Toniann
  Pitassi.
\newblock \href {http://dx.doi.org/10.1109/FOCS.2014.20} {Circuit Complexity,
  Proof Complexity, and Polynomial Identity Testing}.
\newblock In {\em \FOCS{2014}}, pages 110--119, 2014.
\newblock \farXiv{abs/1404.3820}.

\bibitem[KLSS14]{KLSS}
Neeraj Kayal, Nutan Limaye, Chandan Saha, and Srikanth Srinivasan.
\newblock \href {http://dx.doi.org/10.1109/FOCS.2014.15} {{An Exponential Lower
  Bound for Homogeneous Depth Four Arithmetic Circuits}}.
\newblock In {\em \FOCS{2014}}, 2014.
\newblock Pre-print available at \parseECCC{TR14/005}.

\bibitem[KS14]{KS14}
Mrinal Kumar and Shubhangi Saraf.
\newblock \href {http://dx.doi.org/10.1109/FOCS.2014.46} {{On the power of
  homogeneous depth $4$ arithmetic circuits}}.
\newblock In {\em \FOCS{2014}}, 2014.
\newblock Pre-print available at \parseECCC{TR14/045}.

\bibitem[KS15]{KumarSaptharishi15}
Mrinal Kumar and Ramprasad Saptharishi.
\newblock \href {http://eccc.hpi-web.de/report/2015/109/} {An exponential lower
  bound for homogeneous depth-$5$ circuits over finite fields}.
\newblock {\em \ECCC}, 2015.
\newblock \shortECCC{15}{109}.

\bibitem[KST15]{KST15}
Neeraj Kayal, Chandan Saha, and S{\'{e}}bastien Tavenas.
\newblock \href {http://eccc.hpi-web.de/report/2015/181/} {On the size of
  homogeneous and of depth four formulas with low individual degree}.
\newblock {\em \ECCC}, 2015.
\newblock \shortECCC{15}{181}.

\bibitem[NW97]{nw1997}
Noam Nisan and Avi Wigderson.
\newblock \href {http://dx.doi.org/10.1007/BF01294256} {Lower bounds on
  arithmetic circuits via partial derivatives}.
\newblock {\em Computational Complexity}, 6(3):217--234, 1997.
\newblock Available on
  \href{http://citeseerx.ist.psu.edu/viewdoc/summary?doi=10.1.1.90.2644}{\tt
  citeseer:10.1.1.90.2644}.

\bibitem[Raz87]{razborov87}
Alexander~A. Razborov.
\newblock \href {http://dx.doi.org/10.1007/BF01137685} {Lower bounds on the
  size of bounded depth circuits over a complete basis with logical addition}.
\newblock {\em Mathematical notes of the Academy of Sciences of the USSR},
  41(4):333--338, 1987.

\bibitem[Raz06]{Raz06}
Ran Raz.
\newblock \href {http://dx.doi.org/10.4086/toc.2006.v002a006} {Separation of
  Multilinear Circuit and Formula Size}.
\newblock {\em Theory of Computing}, 2(1):121--135, 2006.
\newblock \pFOCS{2004}.
\newblock Pre-print available at \parseECCC{TR04/042}.

\bibitem[Raz09]{raz2004}
Ran Raz.
\newblock \href {http://dx.doi.org/10.1145/1502793.1502797} {Multi-Linear
  Formulas for Permanent and Determinant are of Super-Polynomial Size}.
\newblock {\em Journal of the ACM}, 56(2), 2009.
\newblock \pSTOC{2004}.
\newblock Pre-print available at \parseECCC{TR03/067}.

\bibitem[RY09]{raz-yehudayoff}
Ran Raz and Amir Yehudayoff.
\newblock \href {http://dx.doi.org/10.1007/s00037-009-0270-8} {Lower Bounds and
  Separations for Constant Depth Multilinear Circuits}.
\newblock {\em Computational Complexity}, 18(2):171--207, 2009.
\newblock \pCCC{2008}.
\newblock Pre-print available at \parseECCC{TR08/006}.

\bibitem[Smo87]{smolensky87}
Roman Smolensky.
\newblock \href {http://dx.doi.org/10.1145/28395.28404} {Algebraic Methods in
  the Theory of Lower Bounds for Boolean Circuit Complexity}.
\newblock In {\em \STOC{1987}}, pages 77--82, 1987.

\bibitem[Wil11]{w11}
Ryan Williams.
\newblock {Non-uniform ACC Circuit Lower Bounds}.
\newblock In {\em \CCC{2011}}, pages 115--125, 2011.

\bibitem[Yao85]{yao85}
Andrew Chi-Chih Yao.
\newblock \href {http://dx.doi.org/10.1109/SFCS.1985.49} {Separating the
  polynomial-time hierarchy by oracles}.
\newblock In {\em \FOCS{1985}}, pages 1--10, Oct 1985.

\end{thebibliography}

\appendix

\section{The evaluation perspective on projected shifted partial derivatives}\label{sec:pspd-discussion}

The notion of projected shifted partial derivatives was first introduced by Kayal, Limaye, Saha and Srinivasan \cite{KLSS} in proving lower bounds for homogeneous depth-$4$ circuits. The following is the precise definition they used. 

\begin{definition}[Projected shifted partial derivatives of \cite{KLSS}] Let $k$ and $\ell$ be some parameters. The projected shifted partial derivatives of a polynomial $P(\vecy, \vecz)$, denoted by $\Gamma_{k,\ell}^{\mathrm{PSPD}_0}(P)$, is defined as
\[
\Gamma_{k,\ell}^{\mathrm{PSPD}_0}(P) \spaced{:=} \dim\inbrace{\mathrm{mult}\inparen{\vecz^{=\ell} \partial_{\vecy}^{=k}(P)}}.
\]
where $\mathrm{mult}(f)$ is just the vector of coefficients of all \emph{multilinear} monomials in $f$ in a fixed predefined order. 
\end{definition}

An alternate way to interpret the above definition is to consider the shifted partial derivatives of $P$, and \emph{reduce} them under the relation $x_i^2 = 0$, and only then list the coefficients of the surviving monomials. The rationale for this in \cite{KLSS} was to ensure that non-multilinear terms do not interact with multilinear terms in the shifted partial derivatives of $P$. Hence,
\[
\Gamma_{k,\ell}^{\mathrm{PSPD}_0}(P) \spaced{=} \dim\inbrace{\vecz^{=\ell} \partial_{\vecy}^{=k}(P) \; \mod \setdef{x_i^2}{i\in [n]}}.
\]

Another equally useful definition, which was also employed by Kumar and Saptharishi~\cite{KumarSaptharishi15}, is to reduce the shifted partial derivatives of $P$ with respect to $x_i^2 = x_i$ instead. This also in essence ensures that non-multilinear terms do not interact with the relevant multilinear terms by reducing their degree. We shall denote this by $\Gamma_{k,\ell}^{\mathrm{PSPD}_1}(P)$, which is formally defined to be
\[
\Gamma_{k,\ell}^{\mathrm{PSPD}_1}(P) \spaced{:=} \dim\inbrace{\vecz^{=\ell} \partial_{\vecy}^{=k}(P) \; \mod \setdef{(x_i^2 - x_i)}{i\in [n]}}.
\]
Since any polynomial $f$ has a unique multilinear representation modulo $\setdef{x_i^2 - x_i}{i\in [n]}$, it follows that its evaluations on $\set{0,1}^n$ completely determine the coefficients of the reduced polynomial $f \mod \setdef{x_i^2 - x_i}{i \in [n]}$. Therefore, if $\Gamma_{k,\ell}^{\mathrm{PSPD}}(P)$ is defined as
\[
\Gamma_{k,\ell}^{\mathrm{PSPD}_2}(P) \spaced{:=} \dim\inbrace{\eval_{\set{0,1}^n}\inparen{\vecz^{=\ell} \partial_{\vecy}^{=k}(P)}},
\]
then it follows that 
\[
\Gamma_{k,\ell}^{\mathrm{PSPD}_2}(P) \spaced{=}\Gamma_{k,\ell}^{\mathrm{PSPD}_1}(P). 
\]
Finally, if $P$ was set-multilinear with respect to $\vecx = \vecx_1 \sqcup \cdots \sqcup \vecx_r$ and $\vecy = \vecx_1 \sqcup \cdots \sqcup \vecx_k$, then all partial derivatives of order $k$ with respect to $\vecy$ would be result in polynomials only in $\vecz$. Therefore for such set-multilinear polynomials, 
\begin{eqnarray*}
\Gamma_{k,\ell}^{\mathrm{PSPD}_2}(P) & = &  \dim\inbrace{\eval_{\set{0,1}^n}\inparen{\vecz^{=\ell} \partial_{\vecy}^{=k}(P)}}\\
 & = & \dim\inbrace{\eval_{\set{0}^{n_y} \times \set{0,1}^{n_z}}\inparen{\vecz^{=\ell} \partial_{\vecy}^{=k}(P)}}\\
 & =: & \Gamma_{k,\ell}^{\mathrm{PSPD}}(P) \quad\text{as defined in \autoref{defn:pspd}.}
\end{eqnarray*}
The explicit polynomials for which we shall be show the lower bounds would indeed be set-multilinear and hence there is no loss incurred in restricting to only evaluations on $\set{0}^{n_y} \times \set{0,1}^{n_z}$. \\

\noindent 
For polynomials that are not set-multilinear, clearly
\begin{eqnarray*}
\Gamma_{k,\ell}^{\mathrm{PSPD}_2}(P) & = &  \dim\inbrace{\eval_{\set{0,1}^n}\inparen{\vecz^{=\ell} \partial_{\vecy}^{=k}(P)}}\\
 & \geq & \dim\inbrace{\eval_{\set{0}^{n_y} \times \set{0,1}^{n_z}}\inparen{\vecz^{=\ell} \partial_{\vecy}^{=k}(P)}} \spaced{=:} \Gamma_{k,\ell}^{\mathrm{PSPD}}(P).
\end{eqnarray*}
Hence for the purposes of upper-bounding $\Gamma_{k,\ell}^{\mathrm{PSPD}}()$ for say a term in the circuit computing $P$, taking fewer evaluations only helps.

\end{document}